\documentclass[final]{dmtcs-episciences}
\usepackage{amsmath,amssymb,color,epsfig,bbm,tikz, enumerate,lineno}
\usepackage[round]{natbib}
\usepackage{eurosym}
\newtheorem{theorem}{Theorem}
\newtheorem{lemma}[theorem]{Lemma}
\newtheorem{corollary}[theorem]{Corollary}

\newtheorem{conjecture}{Conjecture}

\newtheorem{question}{Question}
\newenvironment{proof}{{\noin \bf Proof}: }{\qed}
\newcommand{\cc}{\mathcal{C}}

\usepackage{chngcntr}
\counterwithin{theorem}{section}

\author{C.Duffy\affiliationmark{1}\thanks{Support from AARMS}
	\and T.F Lidbetter\affiliationmark{2} \thanks{Support from NSERC (USRA 2015)}
	\and M.E. Messinger\affiliationmark{2} \thanks{Support from NSERC (356119-2011)}
	\and R.J. Nowakowski\affiliationmark{3} \thanks{Support from NSERC (4139-2014)} }

\title[A Variation on Chip-Firing: the diffusion game]{A Variation on Chip-Firing: the diffusion game\thanks{The authors wish to acknowledge the referees for their corrections and suggestions. The clarity of the exposition herein is much improved as a result of the review process.}}

\affiliation{
	Department of Mathematics and Statistics, University of Saskatchewan\\
	Department of Mathematics and Computer Science, Mount Allison University\\
	Department of Mathematics and Statistics, Dalhousie University}
\received{2016-9-20}
\revised{2017-4-19, 2017-12-15}
\accepted{2017-12-19}

\begin{document}

\publicationdetails{20}{2018}{1}{4}{2039}
\maketitle

\begin{abstract}
We introduce a natural variant of the parallel chip-firing game, called the diffusion game.  Chips are initially assigned to vertices of a graph. At every step, all vertices  simultaneously send one chip to each neighbour with fewer chips.  As the dynamics of the parallel chip-firing game occur on a finite set the process is inherently periodic. However the diffusion game is not obviously periodic: even if~\begin{math}2|E(G)|\end{math} chips are assigned to vertices of graph~\begin{math}G\end{math}, there may exist time steps where some vertices have a negative number of chips.  We investigate the process, prove periodicity for a number of graph classes, and pose some questions for future research.\end{abstract}

\section{Introduction}\label{sec:intro}
Each vertex of a graph is assigned a finite integral number of chips. The chips are then redistributed via a stepwise parallel process where at each step all vertices simultaneously send one chip to each neighbour with fewer chips.  If the chips are thought of as molecules, then this process models  \emph{diffusion}: the movement of molecules from regions of high concentration to regions of low concentration.  As a result, the process described above is termed the \emph{ diffusion game}\footnote{We invite the reader to think of the chips as \EUR{1}  coins and consider if the rich always stay rich in this model.}.
 A natural first question is whether the process is necessarily periodic.  

The diffusion game is akin to the \emph{Discharging Method}, wherein a set of charges, rather than chips, are assigned to vertices and then moved around the graph according to a set of rules.  
The Discharging Method is a technique best known for its use in the proof of the famous \emph{ Four Colour Theorem}, but has been used in graph theoretic proofs for over a century.  An overview of this technique is given by \cite{Cranston}.
The diffusion game is also a natural variation of the well-studied \emph{ Chip-Firing Game}.  Chip-firing games on graphs provide a curious area of study for both their mathematical connections and their applications in modeling physical phenomena.  Applications include modeling traffic patterns, the Tutte polynomial, critical groups of graphs, stochastic processes and graph cleaning (see~\cite{cleaningnote,parallel,motivations} and~\cite{firstcleaning}).  In the (sequential) chip-firing game (see~\cite{originalChip}), every vertex in a graph is initially assigned a non-negative integral number of chips.  At each step, one vertex with at least degree-many chips is ``fired''; whereupon it sends one chip to each of its neighbours.   This firing of vertices continues as long as there is a vertex with at least degree-many chips.  The parallel chip-firing game has been studied by~\cite{parallelBitar,parallelKiwi} and \cite{parallelLevine}: at each step, \emph{ every} vertex with at least degree-many chips is fired simultaneously.  It is easy to see that in both the parallel chip-firing game and the diffusion game, the process is completely determined by the initial distribution of chips.  Hence, the term ``game'' is a  misnomer; however we refer to this new variation as the diffusion ``game" because of this close relationship to the parallel chip-firing game.  Since the total number of chips~\begin{math}\cc\end{math} remains constant and it is not possible for a vertex to have a negative number of chips, the dynamics for parallel chip-firing on a graph~\begin{math}G\end{math} occur in the finite set~\begin{math}\{0,1,\dots, \cc\}^{|V(G)|}\end{math} and so parallel chip-firing games eventually exhibit periodic behaviour.  A number of papers have studied the periodicity in parallel chip-firing and results are known about the periods of trees, cycles, complete graphs, and complete bipartite graphs (see~\cite{parallelBitar, DallAsta, parallelKiwi, Kominers, Jiang} and \cite{parallelLevine}). It is important to note that although~\cite{parallelBitar} conjectured that on a graph~\begin{math}G\end{math} any chip configuration is eventually periodic with period at most~\begin{math}|V(G)|\end{math}, \cite{parallelKiwi} provided an example of a graph with period exponential in~\begin{math}|V(G)|\end{math}.

We now formally define the diffusion game.   Let~\begin{math}G\end{math} be a graph with~\begin{math}V(G)~=~\{v_1,v_2,\dots,v_n\}\end{math}.  A chip configuration on~\begin{math}G\end{math} is a vector~\begin{math}c_t = (c_t(v_1), c_t(v_2),\dots,c_t(v_n)) \in \mathbb{Z}^n\end{math}, where~\begin{math}c_0(v_i)\end{math} denotes the number of chips at vertex~\begin{math}v_i\end{math} in the initial configuration and~\begin{math}c_t(v_i)\in\mathbb{Z}\end{math} is the number of chips on vertex~\begin{math}v_i\end{math} at step~\begin{math}t\end{math}.  We define the parallel dynamics as 
\begin{equation}\label{ez}c_{t+1}(v_i) = c_t(v_i) - \Big|\{v_j \in N(v_i)\hspace{0.005in}:\hspace{0.005in} c_t(v_i)>c_t(v_j)\}\Big| + \Big|\{v_j \in N(v_i) \hspace{0.005in}:\hspace{0.005in}c_t(v_i) < c_t(v_j)\}\Big|.\end{equation} 
We refer to this dynamic as \emph{firing}. We note that Equation~(\ref{ez}) is applied synchronously on~\begin{math}V(G)\end{math}.  
It is important to observe that~\begin{math}\sum_{v \in V(G)} c_0(v) =  \sum_{v \in V(G)} c_i(v)\end{math} for all~\begin{math}i \geq 0\end{math}. 
We further observe that the dynamics are unchanged if an additional constant number of chips is added to each vertex. As such we may assume without loss of generality that~\begin{math}\sum_{v \in V(G)} c_i(v) = \cc \geq 0\end{math} for all~\begin{math}i \geq 0\end{math}.

For neighbours~\begin{math}u,v \in V(G)\end{math} such that~\begin{math}c_{k-1}(u) > c_{k-1}(v)\end{math} we say that \emph{in the~\begin{math}k^{th}\end{math} firing~\begin{math}u\end{math} sends a chip to~\begin{math}v\end{math}} and \emph{~\begin{math}v\end{math} receives a chip from~\begin{math}u\end{math}}. When a chip configuration has negative entries the analogy of the process as a redistribution of  chips on a graph no longer holds.
In this  case we may consider a chip configuration to be a function that assigns an integer value to each vertex. Alternatively we may consider entries of the chip configuration to model charge as in \cite{Cranston}. The application of the parallel dynamics in Equation~(\ref{ez}) does not depend on the interpretation of chips as physical objects.  However, if the parallel dynamics necessarily lead to a periodic sequence of chip configurations, we may add a constant so that all~\begin{math}c_i(v) \geq 0\end{math} for all~\begin{math}v \in V\end{math} and all~\begin{math}i \geq 0\end{math} to restore the analogy. It is unknown, however, if the parallel dynamics always lead to a periodic process.

\begin{figure}[htb]
\[ \includegraphics[width=0.5\textwidth]{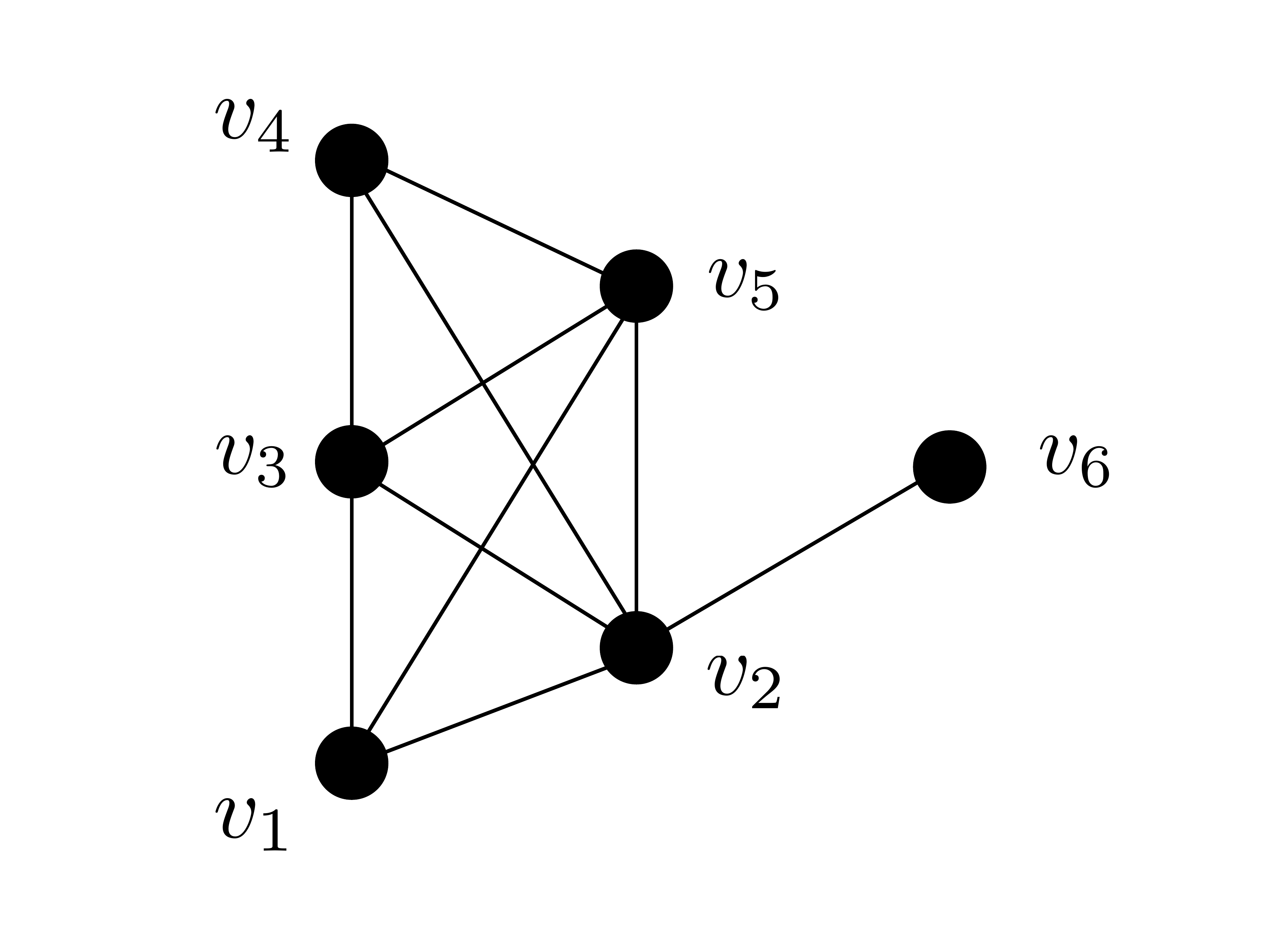}
\includegraphics[width=0.5\textwidth]{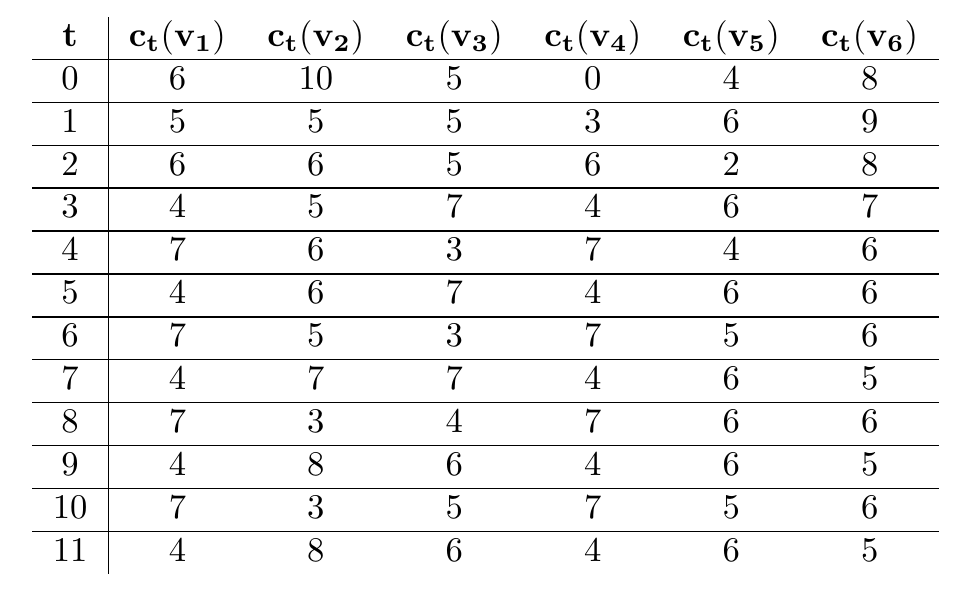} \]

\vspace{-0.25in}\caption{An example of the diffusion game.}
\label{FigXX}
\end{figure} 
Let~\begin{math}c_0\end{math} be an initial chip configuration on some graph~\begin{math}G\end{math}. We observe that if for some~\begin{math}p \in \mathbb{N}\end{math},~\begin{math}c_t = c_{t+p}\end{math}, then the dynamics imply~\begin{math}c_{t+q} = c_{t+p+q}\end{math} for any~\begin{math}q \in \mathbb{N}\end{math}.  We call the minimum such~\begin{math}p\end{math} the \emph{ period length} and given~\begin{math}p\end{math}, we call the minimum such~\begin{math}t\end{math} the \emph{ length of the pre-period}.  In this case, we say that \emph{~\begin{math}c_0\end{math} exhibits eventual periodic behaviour with period length~\begin{math}p\end{math} and pre-period length~\begin{math}t\end{math}}.

We observe that if~\begin{math}c_{t-1} = c_{t}\end{math} for any~\begin{math}t \geq 1\end{math}, then all vertices have the same number of chips (and no chips moved during the~\begin{math}t^{th}\end{math} firing).  Such situations can be thought of as steady-state configurations or stable configurations. We  refer to such configurations as \emph{ fixed}.  If a graph~\begin{math}G\end{math} with an initial configuration~\begin{math}c_0\end{math} eventually proceeds to a configuration~\begin{math}c_t\end{math} that is fixed, we say that~\begin{math}G\end{math} with initial configuration~\begin{math}c_0\end{math} is \emph{eventually fixed}.

As an example, consider the graph shown in Figure~\ref{FigXX} with initial configuration of chips~\begin{math}c_0 = (6,10,5,0,4,8)\end{math} is given in the top row of Figure~\ref{FigXX}.  In the first firing,~\begin{math}v_1\end{math} sends a chip to each of~\begin{math}v_3\end{math} and~\begin{math}v_5\end{math};~\begin{math}v_2\end{math} sends a chip to each of its neighbours;~\begin{math}v_3\end{math} sends a chip to~\begin{math}v_4\end{math} and~\begin{math}v_5\end{math}; and~\begin{math}v_5\end{math} sends a chip to~\begin{math}v_4\end{math}. Observe that since~\begin{math}v_6\end{math} and~\begin{math}v_4\end{math} each are the vertices with least value in their closed neighbourhoods, in the first firing they only receive chips. The resulting configuration is~\begin{math}(5,5,5,3,6,9)\end{math}. Observe that~\begin{math}c_9=c_{11}\end{math}. Thus with initial configuration~\begin{math}c_0 = (6,10,5,0,4,8)\end{math}, the chip configurations on this graph exhibit periodic behaviour with period length~\begin{math}2\end{math} and pre-period length~\begin{math}9\end{math}.

As a second example, consider the graph shown in Figure~\ref{figK}.  For the initial chip configuration, each vertex is assigned degree-many chips.  In the first firing, the vertices of degree~\begin{math}7\end{math} send one chip to each neighbouring leaf.  So at the end of step~\begin{math}1\end{math}, the vertices of degree~\begin{math}7\end{math} each have~\begin{math}3\end{math} chips and the leaves each have~\begin{math}2\end{math} chips.  In the second firing, the vertices of degree~\begin{math}7\end{math} each send one chip to each neighbouring leaf.  Then at the end of step~\begin{math}2\end{math}, the vertices of degree~\begin{math}7\end{math} each have~\begin{math}-1\end{math} chips and the leaves each have~\begin{math}3\end{math} chips.  (In the third firing, the leaves each send a chip to the degree~\begin{math}7\end{math} vertices we find the chip configurations exhibit periodic behaviour with pre-period length~\begin{math}1\end{math} and period length~\begin{math}2\end{math}).  Thus, even with using a total of~\begin{math}2|E(G)|\end{math} chips in the initial configuration, vertices may have a negative number of chips in some steps.  This is very different behaviour from the parallel chip-firing game, where at each step, every vertex has at least~\begin{math}0\end{math} chips and thus, the dynamics occur on a finite set and ensure periodicity.  These observations lead to a fundamental question.

\begin{figure}[htb]
 
\[\includegraphics[width=0.75\textwidth]{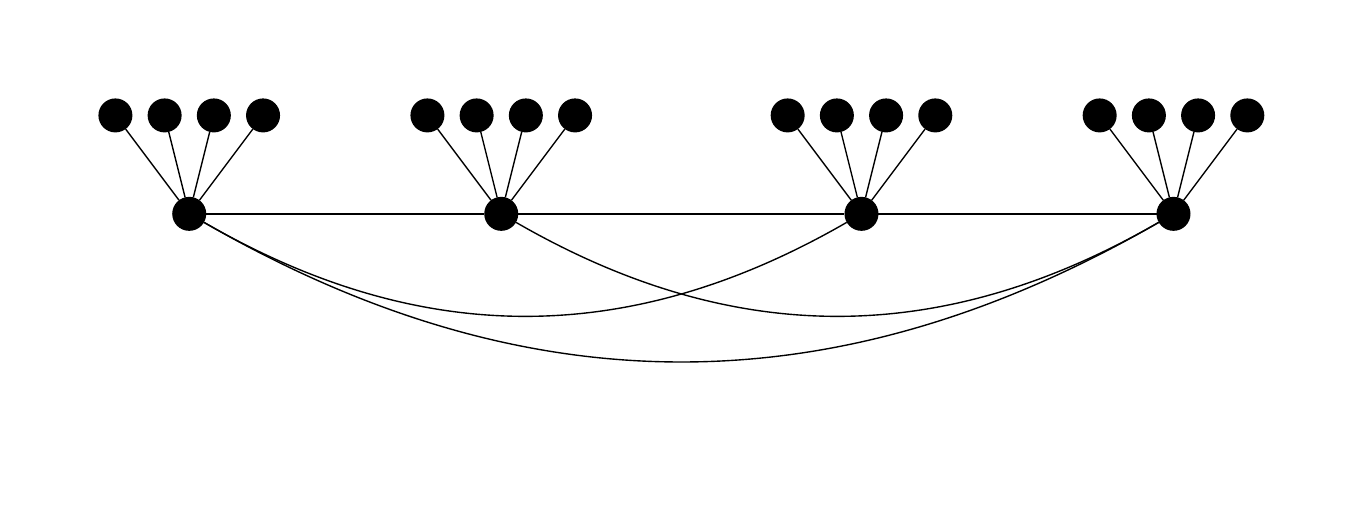}\]

\vspace{-0.4in}\caption{An example of a graph that has a chip configuration leading to vertices with negative chips.}

\label{figK}

\end{figure}

\begin{question}\label{q1} Let~\begin{math}G\end{math} be a finite  graph with an initial configuration.  Do the chip configurations on~\begin{math}G\end{math} eventually exhibit periodic behaviour?\end{question} 

Simulation of the process on a~\begin{math}50 \times 100\end{math} Cartesian grid where each vertex had a random initial value of between~\begin{math}1\end{math} and~\begin{math}200\end{math} chips gave a process that was eventually periodic with period~\begin{math}2\end{math} in each of~\begin{math}200\end{math} simulations.  In fact, every configuration and graph the authors considered was fixed or has period length~\begin{math}2\end{math} after some pre-period.  With this in mind, we say that a graph, together with an initial configuration is \emph{ tight} if the resulting process is eventually (\textit{i.e.,} after some pre-period) fixed or periodic with period length~\begin{math}2\end{math}.

\begin{conjecture}\label{conj} Every finite graph and initial configuration is tight. \end{conjecture}

In Section~\ref{sec:finitePeriod}, we partially answer Question~\ref{q1} affirmatively for particular families of graphs. 
As evidence for Conjecture~\ref{conj},  in Section~\ref{sec:period2}, we consider some natural initial distributions of chips and show that some families of graphs do eventually exhibit periodic behaviour with period length~\begin{math}2\end{math}.  Finally, in Section~\ref{sec:questions}, we conclude with a brief exploration of some of the open  questions that remain.
 
Given any finite graph and fixed number~\begin{math}\mathcal{C}\end{math} of chips, we can consider an auxiliary directed graph~\begin{math}A(G,\mathcal{C})\end{math} that encodes the chip configurations.  Each vertex of~\begin{math}A(G,\mathcal{C})\end{math} is a chip configuration where the total number of chips is~\begin{math}\mathcal{C}\end{math}.  For a pair of vertices~\begin{math}u,v \in V(A(G,\mathcal{C}))\end{math}, there is an arc from~\begin{math}u\end{math} to~\begin{math}v\end{math} if applying the dynamics of the chip configuration corresponding to~\begin{math}u\end{math} yields the configuration at~\begin{math}v\end{math}.  Observe that since the process is deterministic, the out-degree of any vertex is~\begin{math}1\end{math}, but the in-degree can be arbitrary.  Consider~\begin{math}P_3\end{math}, for example, and let~\begin{math}c_{t} = (c_t(x),c_t(y),c_t(z))\end{math} where~\begin{math}\deg(y)=2\end{math}.  If~\begin{math}c_t = (1,1,2)\end{math} then~\begin{math}c_{t-1}\end{math} could have been~\begin{math}(0,1,3)\end{math},~\begin{math}(0,2,2)\end{math}, or~\begin{math}(0,3,1)\end{math}; thus in the auxiliary graph~\begin{math}A(P_3,4)\end{math}, the vertex corresponding to configuration~\begin{math}(1,1,2)\end{math} has three parent vertices.  Note that Question~\ref{q1} is equivalent to asking if every directed path in~\begin{math}A(G,\mathcal{C})\end{math} is finite.  Conjecture~\ref{conj} is equivalent to saying that each strongly connected component of~\begin{math}A(G,\mathcal{C})\end{math} has cardinality~\begin{math}1\end{math} or~\begin{math}2\end{math}.

\section{General Periodicity Results}\label{sec:finitePeriod}

In this section, we affirm Question~\ref{q1} for some  well-known classes of graphs.  To do this, we show that the values for~\begin{math}c_t = (c_t(v_1),c_t(v_2),\dots,c_t(v_n))\end{math} are in a finite set.  In other words, we show for every integer~\begin{math}t \geq 0\end{math}, the number of chips on any vertex in such graphs is always in a finite interval. As there will be a finite number of distinct configurations in the process and~\begin{math}c_{t+1}\end{math} is determined by~\begin{math}c_t\end{math}, and as the diffusion game continues indefinitely, the configurations must eventually exhibit periodic behaviour. This approach will be used for paths, cycles, wheels, complete graphs, and complete bipartite graphs.  We begin with a few definitions and preliminary results in Section~\ref{sub:pre}.

\subsection{Definitions and Preliminary Results}\label{sub:pre}

Let~\begin{math}G\end{math} be a graph with initial configuration~\begin{math}c_0\end{math}.  With respect to~\begin{math}c_0\end{math}, we say that~\begin{math}v \in V(G)\end{math} is \emph{ bounded above} if~\begin{math}\mbox{there exists }~M\in\mathbb{Z}\end{math} such that for every integer~\begin{math}t \geq 0\end{math}, we have~\begin{math}c_t(v) \leq M\end{math}. Similarly, with respect to~\begin{math}c_0\end{math}, we say that~\begin{math}v\end{math} is \emph{ bounded below} if~\begin{math}\mbox{there exists }~m \in \mathbb{Z}\end{math} such that for every integer~\begin{math}t \geq 0\end{math}, we have~\begin{math}c_t(v) \geq m\end{math}.  If~\begin{math}v\end{math} is bounded both above and below, then~\begin{math}v\end{math} is \emph{ bounded}; otherwise~\begin{math}v\end{math} is \emph{ unbounded}.

\begin{lemma} \label{lem:boundedvalue} Let~\begin{math}G\end{math} be a graph with an initial configuration.  Every vertex in~\begin{math}G\end{math} is bounded if and only if the chip configurations eventually exhibit periodic behaviour.  
\end{lemma}

\begin{proof} Let~\begin{math}G\end{math} be a graph with initial configuration~\begin{math}c_0\end{math}.  If every vertex is bounded,~\begin{math}\mbox{there exists }\end{math}~\begin{math}M \in \mathbb{Z}^+\end{math} such that~\begin{math}-M \leq c_t(v) \leq M\end{math} for all~\begin{math}v \in V(G)\end{math}.  Thus, the values for~\begin{math}c_t\end{math} are in a finite set and so the chip configurations must eventually exhibit periodic behaviour.

For the other direction, suppose vertex~\begin{math}v\end{math} is not bounded above.  Then for each~\begin{math}M \in \mathbb{Z}\end{math},~\begin{math}\mbox{there exists }\end{math}~\begin{math}t \geq 0\end{math} such that~\begin{math}c_t(v) > M\end{math}.  Since~\begin{math}v\end{math} takes on an infinite number of positive values, the chip configurations will never be periodic.\end{proof}

In fact, the chip configurations will eventually be periodic if there is a bound on the chip difference between every pair of adjacent vertices. We use the notation~\begin{math}u\sim v\end{math} to indicate that vertices~\begin{math}u\end{math} and~\begin{math}v\end{math} are adjacent.

\begin{lemma} \label{lem:boundededge} Let~\begin{math}G\end{math} be a graph with an initial configuration.  If there exists a constant~\begin{math}M\end{math} such that~\begin{displaymath}|c_t(u) -c_t(v)| \leq M\end{displaymath} for all~\begin{math}u \sim v\end{math} and all~\begin{math}t \geq 0\end{math}, then the chip configurations will eventually exhibit periodic behaviour. \end{lemma}

\begin{proof} Let~\begin{math}G=(V,E)\end{math} be a graph with an initial chip configuration and let~\begin{math}|V(G)|~=~n\end{math}.  Assume, without loss of generality, that~\begin{math}\sum_{v \in V(G)} c_0(v) = \cc \geq 0\end{math} and recall that the total number of chips,~\begin{math}\cc\end{math}, does not change when the parallel dynamics are applied.  Clearly for all~\begin{math}t \geq 0\end{math}, 
	\begin{equation}min_{v \in V} c_t(v) \leq \cc/n \leq \max_{v \in V} c_t(v) \leq \min_{v \in V} c_t(v)+diam(G) \cdot M,	\end{equation}
	where~\begin{math}diam(G)\end{math} is the diameter of~\begin{math}G\end{math} and~\begin{math}M\end{math} is a constant such that~\begin{math}|c_t(u)-c_t(v)| \leq M\end{math} for all neighbours~\begin{math}u,v\end{math} and~\begin{math}t \geq 0\end{math}.  It follows that 
	\begin{equation}
	c_t(v)| \leq \max\left\lbrace \left|diam(G)\cdot M+\frac{\cc}{n}\right|,\left|\frac{\cc}{n}-diam(G)\cdot M\right|\right\rbrace
	\end{equation}
	 for all~\begin{math}t \geq 0\end{math} and for all~\begin{math}v\end{math}.  The result now follows from Lemma~\ref{lem:boundedvalue}.

\end{proof}

We next prove a useful implication of having an unbounded vertex.

\begin{lemma}\label{lem:1AdjUnbounded}  Let~\begin{math}G\end{math} be a graph with an initial chip configuration.  If~\begin{math}v\end{math} is unbounded, then~\begin{math}\mbox{there exists }~u \in N(v)\end{math} that is also unbounded. \end{lemma}

\begin{proof}   Let~\begin{math}G\end{math} be a graph with an initial chip configuration.  Suppose~\begin{math}v \in V(G)\end{math} is unbounded above and, for a contradiction, assume all vertices in~\begin{math}N(v)=\{v_1,v_2, \dots, v_{\deg(v)}\}\end{math} are bounded.  Then for every~\begin{math}i \in \{1,2,\dots, \deg(v)\}\end{math}, there exists a constant~\begin{math}\cc_i\end{math} such that~\begin{math}v_i\end{math} is bounded above by~\begin{math}\cc_i\end{math}.  

Let~\begin{math}\cc_m=\max\{\cc_1,\cc_2,\dots, \cc_{\deg(v)}\}\end{math}. If~\begin{math}c_t(v)>\cc_m\end{math}  for some~\begin{math}t \geq 0\end{math}, then~\begin{displaymath}c_{t+1}(v)=c_t(v)-\deg(v) < c_t(v).\end{displaymath}  If~\begin{math}c_t(v)\leq \cc_m\end{math} for some~\begin{math}t \geq 0\end{math}, then~\begin{displaymath}c_{t+1}(v)\leq \cc_m-1+\deg(v).\end{displaymath} Thus, the number of chips at~\begin{math}v\end{math} cannot exceed~\begin{math}\max\{c_1(v),\cc_m-1+\deg(v)\}\end{math}.  Therefore if~\begin{math}v\end{math} is unbounded,~\begin{math}v\end{math} must have an unbounded neighbour.\end{proof}


\subsection{Cycles, Paths, and Wheels}


In this section, we first show that for any finite initial chip configuration on~\begin{math}C_n\end{math} or~\begin{math}P_n\end{math} with~\begin{math}n\geq 3\end{math}, every vertex is bounded, and hence the chip configurations will eventually be periodic (\textit{i.e.,} will be periodic after some pre-period of steps). 

\begin{lemma}\label{lem:deg2diff}  
	Let~\begin{math}G\end{math} be a graph with an initial chip configuration and~\begin{math}u,v\end{math} be adjacent vertices in~\begin{math}G\end{math} with~\begin{math}\deg(u)=2\end{math},~\begin{math}\deg(v) \in \{1,2\}\end{math}.  For any integer~\begin{math}t \geq 0\end{math}, \begin{displaymath}|c_{t+1}(u)-c_{t+1}(v)| \leq \max\Big\{ 3, |c_t(u) -c_t(v)|\Big\}.\end{displaymath}
  \end{lemma} 

\begin{proof}
	Let~\begin{math}u,v\end{math} be adjacent vertices in graph~\begin{math}G\end{math} where~\begin{math}\deg(u)=2\end{math} and~\begin{math}\deg(v) \in \{1,2\}\end{math}.  If~\begin{math}|c_t(u)~-~c_t(v)|~<~2\end{math}, then it is easy to see that~\begin{math}|~c_{t+1}~(u)~-~c_{t+1}~(v)~|~\leq~3\end{math}.
	Suppose~\begin{math}c_t(u)-c_t(v)\geq 2\end{math}.  As~\begin{math}c_t(u) > c_t(v)\end{math}, vertex~\begin{math}u\end{math} will send a chip to~\begin{math}v\end{math} and since~\begin{math}\deg(u)=2\end{math}, vertex~\begin{math}u\end{math} sends or receives at most one additional chip.  Since~\begin{math}\deg(v)\in \{1,2\}\end{math}, vertex~\begin{math}v\end{math} sends or receives at most one additional chip.  Then~\begin{math}c_{t+1}(u)\end{math} will equal one of:~\begin{math}c_t(u)\end{math},~\begin{math}c_t(u)-1\end{math},~\begin{math}c_t(u)-2\end{math}; and~\begin{math}c_{t+1}(v)\end{math} will equal one of~\begin{math}c_t(v)\end{math},~\begin{math}c_t(v)+1\end{math},~\begin{math}c_t(v)+2\end{math}.  By considering the possible combinations, the result follows.  If~\begin{math}c_t(v)-c_t(u) \geq 2\end{math}, a similar argument yields the result.\end{proof}

\begin{theorem}\label{thm:PnCnPeriodic} For any finite initial chip configuration on~\begin{math}P_n\end{math} or~\begin{math}C_n\end{math} with~\begin{math}n\geq 3\end{math}, the chip configurations will eventually exhibit periodic behaviour. \end{theorem}
	
\begin{proof} Let~\begin{math}M_0 = |\max_{v \in V(G)} \{c_0(v)\}|\end{math} and~\begin{math}N_0 = |\min_{v \in V(G)} \{c_0(v)\}|\end{math}.  By Lemma \ref{lem:deg2diff}, for any adjacent vertices~\begin{math}u\end{math} and~\begin{math}v\end{math}, \begin{equation}|c_{t+1}(u)-c_{t+1}(v)| \leq \max\Big\{ 3, |c_t(u) -c_t(v)|\Big\} \leq \max\Big\{ 3, |c_0(u) -c_0(v)|\Big\} \leq \max\Big\{ 3,N_0+M_0\Big\}\end{equation} for all integers~\begin{math}t \geq 0\end{math}.  The result now follows from Lemma \ref{lem:boundededge}.\end{proof}

We next show that for any finite initial chip configuration on a wheel graph, the configurations will eventually be periodic.  Let~\begin{math}W_n\end{math} denote the wheel graph on~\begin{math}n\end{math} vertices.  

\begin{lemma}\label{lem:wheelDiff1} For any finite initial chip configuration on~\begin{math}W_n\end{math} with~\begin{math}n \geq 4\end{math}, let~\begin{math}u,v\end{math} be adjacent vertices of degree~\begin{math}3\end{math}.  For any integer~\begin{math}t \geq 0\end{math}, 

(a) if~\begin{math}|c_t(u) - c_t(v)| < 3\end{math}, then~\begin{math}|c_{t+1}(u)-c_{t+1}(v)| \leq 6\end{math} and 

(b) if~\begin{math}|c_t(u) - c_t(v)| \geq 3\end{math}, then~\begin{math}|c_{t+1}(u) - c_{t+1}(v)| \leq |c_t(u) - c_t(v)|\end{math}.  \end{lemma}
	
\begin{proof} 
	Let~\begin{math}u,v\end{math} be adjacent vertices of degree~\begin{math}3\end{math}  and let~\begin{math}w\end{math} be the universal vertex in~\begin{math}W_n\end{math}.
	We claim the difference between~\begin{math}|c_t(u)-c_t(v)|\end{math} and~\begin{math}|c_{t+1}(u)-c_{t+1}(v)|\end{math} is at most~\begin{math}4\end{math}.
	Without loss of generality, assume~\begin{math}c_t(u) > c_t(v)\end{math}. Let~\begin{math}c_t(v) = c_t(u) + \ell\end{math}.
	Since~\begin{math}u\end{math} necessarily sends one chip to~\begin{math}v\end{math} in this round, we have~\begin{math}c_{t+1}(u) \leq c_t(u) + 2\end{math} and~\begin{math}c_{t+1}(v) \geq c_t(v)  -2\end{math}. Therefore~\begin{math}c_{t+1}(u) - c_{t+1}(v) \leq \ell + 4\end{math}. From this it follows that the difference between~\begin{math}|c_t(u)-c_t(v)|\end{math} and~\begin{math}|c_{t+1}(u)-c_{t+1}(v)|\end{math} is at most~\begin{math}4\end{math}.
	Thus, (a) holds.

For some~\begin{math}t \geq 0\end{math}, suppose without loss of generality that~\begin{math}c_t(u) - c_t(v) \geq 3\end{math}.  During the~\begin{math}t+1^{th}\end{math} firing, a chip is sent from~\begin{math}u\end{math} to~\begin{math}v\end{math} and we have
\begin{eqnarray*} c_{t+1}(u) \in& \Big\{c_t(u)-3, c_t(u)-2, c_t(u)-1, c_t(u), c_t(u)+1\Big\}, \mbox{and} \\ c_{t+1}(v) \in& \Big\{c_t(v)-1, c_t(v), c_t(v)+1, c_t(v)+2, c_t(v)+3\Big\}.\end{eqnarray*} 

 Then~\begin{math}c_{t+1}(u) = c_t(u)+x\end{math} and~\begin{math}c_{t+1}(v)=c_t(v)+y\end{math},  where~\begin{math}x \in \{-3,-2,-1,0,1\}\end{math} and~\begin{math}y \in \{-1,0,1,2,3\}\end{math}.  
 We observe that~\begin{math}y-x \leq 6\end{math}.  
 If~\begin{math}|c_{t+1}(u)-c_{t+1}(v)| = c_{t+1}(v) - c_{t+1}(u)\end{math}, we have 
 \begin{eqnarray*} y-x \leq 6 \leq  2\Big( c_t(u)-c_t(v)\Big)~&\Longleftrightarrow&~ \Big(c_t(v)+y\Big) -\Big(c_t(u)+x\Big) \leq c_t(u) - c_t(v) \\ &\Longleftrightarrow&~ \Big|c_{t+1}(u)-c_{t+1}(v)\Big| \leq \Big|c_t(u)-c_t(v)\Big|.\end{eqnarray*}  
 Otherwise, 
 \begin{displaymath}|c_{t+1}(u)-c_{t+1}(v)| = c_{t+1}(u)-c_{t+1}(v)  = c_t(u)-c_t(v)+x-y   \leq |c_t(u)-c_t(v)|.\end{displaymath} 
 The last inequality results from the fact that~\begin{math}x-y \leq 0\end{math}.  To see this, suppose~\begin{math}x > y\end{math}.  Then~\begin{math}x \in \{0,1\}\end{math} and~\begin{math}c_t(w) \geq c_t(u) \geq c_t(v)\end{math}, which implies~\begin{math}c_{t+1}(v) \geq c_t(v)+1\end{math} and~\begin{math}y \in \{1,2,3\}\end{math}. This contradicts the assumption~\begin{math}x > y\end{math}.\end{proof}

\begin{theorem}\label{thm:WnPeriodic} For any finite initial chip distribution on~\begin{math}W_n\end{math} with~\begin{math}n\geq 4\end{math}, the chip configurations will eventually exhibit periodic behaviour. \end{theorem}

\begin{proof} We  show  the number of chips on neighbouring vertices is bounded at each step.  Label the vertices of~\begin{math}W_n\end{math} as~\begin{math}w\end{math},~\begin{math}v_0\end{math},~\begin{math}v_1\end{math}, \dots,~\begin{math}v_{n-2}\end{math} where~\begin{math}w \sim v_i\end{math} for~\begin{math}0 \leq i \leq n-2\end{math} and~\begin{math}v_i \sim v_{i+1 \mod n-1}.\end{math}  

Assume~\begin{math}W_n\end{math} contains an unbounded vertex and observe that by Lemma~\ref{lem:1AdjUnbounded}, there must be an unbounded vertex of degree~\begin{math}3\end{math}: let~\begin{math}v_0\end{math} be an unbounded vertex of degree~\begin{math}3\end{math}.  Further, let \begin{equation}R_i = \begin{cases} \max \Big\{ |c_0(v_i)-c_0(v_{i+1})|, 6 \Big\} & \text{ for } 0 \leq i \leq n-3 \\  \max \Big\{ |c_0(v_{n-2})-c_0(v_0)|, 6 \Big\} & \text{ for } i=n-2\\ \end{cases}.\end{equation}
 By Lemma~\ref{lem:wheelDiff1}, at any step~\begin{math}t\end{math}, the difference between the number of chips at~\begin{math}v_i\end{math} and~\begin{math}v_{i+1}\end{math} cannot exceed~\begin{math}R_i\end{math}.  Certainly, for any~\begin{math}M \in \mathbb{Z}^+\end{math},  if~\begin{math}v_0\end{math} has~\begin{math}M\end{math} chips, then the number of chips at each other vertex of degree~\begin{math}3\end{math} is at least \begin{equation*}\label{eq:R}M-\sum_{j=0}^{n-2}R_j.\end{equation*} We next show that the difference between the number of chips on~\begin{math}w\end{math} and~\begin{math}v_0\end{math} can never exceed 
 
\begin{equation}\label{pp}R_w=\max\left\{ n+2+ \sum_{j=0}^{n-2}R_j, \max_{v_i\in V(W_n)}\left\{\left|c_0(w)-c_0(v_i)\right|\right\}\right\}.\end{equation} 

For some step~\begin{math}t \geq 0\end{math}, suppose that~\begin{math}|c_t(w)-c_t(v_0)|> \sum_{j=0}^{n-2}R_j\end{math}. Since by~(\ref{pp}),~\begin{math}|c_t(v_0)-c_t(v_i)| \leq \sum_{j=0}^{n-2} R_j\end{math} for~\begin{math}i \in \{1,2,\dots,n-2\}\end{math}, this implies~\begin{math}w\end{math} has fewer (or greater) chips than all~\begin{math}v_i\end{math} for~\begin{math}0 \leq i \leq n-2\end{math}.  During the next firing,~\begin{math}w\end{math}  receives (or sends)~\begin{math}n-1\end{math} chips;  the difference between the number of chips on~\begin{math}v_i\end{math} and~\begin{math}w\end{math} has not increased, so~\begin{math}|c_{t+1}(w)-c_{t+1}(v_0)|\leq |c_t(w)-c_t(v_0)| \leq R_w\end{math} and~(\ref{pp}) holds.
	
Next, for some step~\begin{math}t \geq 0\end{math}, suppose that~\begin{math}|c_t(w)-c_t(v_0)|\leq \sum_{j=0}^{n-2}R_j\end{math}.  In the worst case,~\begin{math}w\end{math} receives (or sends)~\begin{math}n-1\end{math} chips and~\begin{math}v_0\end{math} sends (or receives) one chip to (or from) each of its~\begin{math}3\end{math} neighbours.  Thus, \begin{equation}|c_{t+1}(w)-c_{t+1}(v_0)| \leq (n-1)+3+\sum_{j=0}^{n-2}R_j \leq R_w\end{equation}  and~(\ref{pp}) holds. Therefore, for all~\begin{math}t \geq 0\end{math},~\begin{math}|c_t(w)-c_t(v_0)|\leq R_w.\end{math}  Therefore by Lemma~\ref{lem:boundededge}, the chip configurations will eventually exhibit periodic behaviour.  By Lemma~\ref{lem:boundedvalue}, we have a contradiction.\end{proof}


\subsection{Complete and Complete Bipartite Graphs}


For a vertex~\begin{math}v\end{math} in graph~\begin{math}G\end{math}, let~\begin{math}N[v]\end{math} denote the closed neighbourhood of vertex~\begin{math}v\end{math}; that is,~\begin{math}N[v] = N(v) \cup \{v\}\end{math}.

\begin{lemma}\label{lem:neighbourhood} Let~\begin{math}G\end{math} be a graph with an initial configuration.  For any~\begin{math}u,v \in V(G)\end{math}, if~\begin{math}N(u)=N(v)\end{math} or~\begin{math}N[u]=N[v]\end{math}, then for every integer~\begin{math}t \geq 0\end{math}, \begin{displaymath}|c_t(u)-c_t(v)|\leq\max\Big\{|c_0(u)-c_0(v)|,2\deg(u)\Big\}.\end{displaymath}\end{lemma}

\begin{proof}  
	
	Let~\begin{math}G\end{math} be a graph with an initial configuration and let~\begin{math}u,v \in V(G)\end{math} with~\begin{math}N(u)=N(v)\end{math} or~\begin{math}N[u]=N[v]\end{math}; then~\begin{math}\deg(u)=\deg(v)\end{math}.  Since we do not know whether~\begin{math}u\end{math} and~\begin{math}v\end{math} are adjacent, we  consider~\begin{math}N(u)\setminus\{v\}\end{math}.  Next, without loss of generality, suppose~\begin{math}c_t(u) \geq c_t(v)\end{math} and let \vspace{0.1in}

\begin{itemize}
	\item~\begin{math}x_1 = |\{w \in N(u)\backslash \{v\}~:~ c_t(w) > c_t(u)\}|\end{math}
	
\item~\begin{math}x_2 = |\{w \in N(u)\backslash \{v\}~:~ c_t(w) < c_t(v)\}|\end{math}
	
\item~\begin{math}y_1 = |\{w \in N(u)\backslash \{v\}~:~ c_t(v) =c_t(w)< c_t(u)\}|\end{math}
	
\item~\begin{math}y_2 = |\{w \in N(u)\backslash \{v\}~:~ c_t(v) <c_t(w)=c_t(u)\}|\end{math}
	
\item~\begin{math}y_3 = |\{w \in N(u)\backslash \{v\}~:~ c_t(v) <c_t(w)< c_t(u)\}|\end{math}

\item~\begin{math}\begin{displaystyle}
z=
\begin{cases} 
1 & \text{ if } u\sim v \text{ and } c_t(u)>c_t(v);\\
0 & \text{ otherwise.}\\ 
\end{cases}\end{displaystyle}\end{math}
\vspace{0.1in}

\end{itemize}
\noindent 
Then,
\begin{eqnarray*}c_{t+1}(u)-c_{t+1}(v) &=& \Big(c_t(u)+x_1-x_2-y_1-y_3 - z\Big)-\Big(c_t(v)+x_1-x_2+y_2+y_3+z\Big) \\ &=& c_t(u)-c_t(v) -(y_1+y_2+2y_3+2z) \\ &\leq& \max\{c_t(u)-c_t(v), 2\deg(u)\}\end{eqnarray*} 
as~\begin{math}0 \leq y_1+y_2+2y_3+2z \leq 2\deg(u)\end{math}.

\end{proof}

Corollary~\ref{cor:equal} follows immediately from the proof of Lemma~\ref{lem:neighbourhood}, as if~\begin{math}c_t(u)=c_t(v)\end{math} then~\begin{math}y_1=y_2=y_3=0\end{math}.

\begin{corollary}\label{cor:equal} Let~\begin{math}G\end{math} be a graph with an initial configuration and let~\begin{math}u,v\in V(G)\end{math} where~\begin{math}N(u)=N(v)\end{math}.  If~\begin{math}c_t(u)=c_t(v)\end{math} for some~\begin{math}t \geq 0\end{math}, then~\begin{math}c_r(u)=c_r(v)\end{math} for all integers~\begin{math}r\geq t\end{math}.  \end{corollary}

\begin{theorem}\label{thm:KnPeriodic} For any finite initial configuration of chips on~\begin{math}K_n\end{math} for~\begin{math}n \geq 2\end{math}, the configurations will eventually be periodic. \end{theorem}

\begin{proof} Let~\begin{math}M = |\max_{v \in V(K_n)} \{c_0(v)\}|\end{math} and~\begin{math}N = |\min_{v \in V(K_n)} \{c_0(v)\}|\end{math}.  By Lemma \ref{lem:neighbourhood}  for all~\begin{math}u,v \in V(K_n)\end{math} and all~\begin{math}t \geq0\end{math} \begin{equation}|c_t(u)-c_t(v)|\leq\max\Big\{|c_0(u)-c_0(v)|,2\deg(u)\Big\} \leq \max \{M+N, 2(n-1)\}.\end{equation}
The result now follows from Lemma \ref{lem:boundededge}. \end{proof}

\begin{theorem} For any finite initial configuration of chips on~\begin{math}K_{m,n}\end{math} for~\begin{math}m,n\geq 1\end{math}, the configurations will eventually be periodic.\end{theorem}

\begin{proof}  Let~\begin{math}X\end{math} and~\begin{math}Y\end{math} partition the vertices of~\begin{math}K_{m,n}\end{math} where~\begin{math}|X|=m\end{math},~\begin{math}|Y|=n\end{math}.  For a contradiction, assume that there is an unbounded vertex~\begin{math}x \in X\end{math}. Without loss of generality, assume that~\begin{math}x\end{math} is unbounded above. Then for any~\begin{math}M\in\mathbb{Z}^+\end{math} there exists a \emph{ least} integer~\begin{math}t\end{math} such that~\begin{math}c_t(x)\geq M\end{math}.  As the only vertices adjacent to~\begin{math}x\end{math} are in~\begin{math}Y\end{math}, at least one chip must  be sent from vertices in~\begin{math}Y\end{math} to~\begin{math}x\end{math} during the~\begin{math}t^{th}\end{math} firing:  suppose~\begin{math}x\end{math} receives a chip from vertex~\begin{math}y\in Y\end{math}.  Then as~\begin{math}c_{t-1}(x)\geq M-\deg(x)\end{math}, we find~\begin{math}c_{t-1}(y)\geq M-\deg(x)+1\end{math}.

By Lemma~\ref{lem:neighbourhood},  \begin{equation}c_{t-1}(x') \geq M - \deg(x)-\max\Big\{ |c_0(x)-c_0(x')|, 2|Y|\Big\}\end{equation}  for each~\begin{math}x' \in X\backslash \{x\}\end{math} and \begin{equation}c_{t-1}(y') \geq M-\deg(x)+1-\max\Big\{|c_0(y)-c_0(y')|,2|X|\Big\}\end{equation} for each~\begin{math}y' \in Y\backslash\{y\}.\end{math}  As~\begin{math}t\end{math} is the least integer for which~\begin{math}c_t(u)\geq M\end{math}, the minimum total number of chips on~\begin{math}K_{m,n}\end{math} after the~\begin{math}t-2^{th}\end{math} firing must be at least  \begin{eqnarray}\label{k}&\Big(M-\deg(x)\Big) + \sum_{x'\in X\backslash\{x\}}\Big(M-\deg(x)-\max\{|c_0(x)-c_0(x')|,2|Y|\}\Big)+\nonumber\\ &\Big(M-\deg(x)+1\Big)+ \sum_{y'\in Y\backslash\{y\}}\Big(M-\deg(x)+1-\max\{|c_0(y)-c_0(y')|,2|X|\}\Big).~\end{eqnarray}  

Let~\begin{math}\cc =\sum_{v \in V(G)} c_0(v)\end{math}.
By  choosing~\begin{math}M\end{math} sufficiently large so that the sum in~(\ref{k}) exceeds~\begin{math}\cc\end{math}, we arrive at a contradiction. Therefore~\begin{math}x\end{math} is bounded above.  As~\begin{math}x\end{math} was chosen arbitrarily it follows that every vertex must be bounded above.  The result now follows from Lemma \ref{lem:boundedvalue}. \end{proof}


\section{Period Length 2}\label{sec:period2}


\subsection{Preliminary Results on Period Length 2}

Let~\begin{math}G\end{math} be a graph and~\begin{math}c_i\end{math} be a configuration. For all~\begin{math}u \in V(G)\end{math} let 
\begin{align*}
\Delta_i^-(u) &= \left| \{v : v \in N(u) \mbox{ and } c_i(u) > c_i(v)\} \right|,\\
\Delta_i^+(u) &= \left| \{v : v \in N(u) \mbox{ and } c_i(u) < c_i(v)\} \right|, \mbox{ and}\\
\Delta_i(u) &= \Delta_i^+(u) - \Delta_i^-(u).
\end{align*}
Observe~\begin{math}c_{i+1}(u) = c_i(u) + \Delta_i(u)\end{math}.

We say that a graph~\begin{math}G\end{math} with configuration~\begin{math}c_i\end{math} has \emph{property plus} if 

\begin{enumerate}
	\item~\begin{math}c_i(u) + \Delta_i(u) > c_i(v) + \Delta_i(v)\end{math} for all~\begin{math}uv \in E(G)\end{math}, where~\begin{math}c_i(u) < c_i(v)\end{math}, and
	\item~\begin{math}c_i(u) + \Delta_i(u) = c_i(v) + \Delta_i(v)\end{math} for all~\begin{math}uv \in E(G)\end{math}, where~\begin{math}c_i(u) = c_i(v)\end{math}.
\end{enumerate}

If each vertex has a positive number of chips, then the first requirement of property plus may be interpreted as~\begin{math}u\end{math} sending a chip to~\begin{math}v\end{math} in round~\begin{math}i\end{math} and~\begin{math}v\end{math} returning that chip to~\begin{math}u\end{math} in the subsequent round. From property plus we find a necessary and sufficient condition for a graph and initial configuration to be tight.

\begin{theorem}\label{thm:propertyplus}
	A graph~\begin{math}G\end{math} with initial configuration~\begin{math}c_0\end{math} is tight if and only if there exists~~\begin{math}t\end{math} such that~\begin{math}c_t\end{math} has property plus.
\end{theorem}

\begin{proof}
	If~\begin{math}n < 3\end{math}, then the result is trivial. As such, it suffices to consider~\begin{math}n \geq 3\end{math}. Assume~\begin{math}c_t\end{math} has property plus. It suffices to show~\begin{math}c_t(u) = c_{t+2}(u)\end{math} for all~\begin{math}u \in V(G)\end{math}. Since~\begin{math}c_t\end{math} has property plus we observe~\begin{math}\Delta_t^+(u) = \Delta_{t+1}^-(u)\end{math} and~\begin{math}\Delta_t^-(u) = \Delta_{t+1}^+(u)\end{math}. Together these two facts imply~\begin{math}\Delta_t(u)  = -\Delta_{t+1}(u)\end{math}. From this we see \begin{equation}c_{t+2}(u) = c_{t+1}(u) + \Delta_{t+1}(u) = c_t(u) + \Delta_{t}(u) + \Delta_{t+1}(u) = c_t(u).\end{equation}
	
	Assume now that~\begin{math}G\end{math} with initial configuration~\begin{math}c_0\end{math} is eventually tight. Further, assume without loss of generality that~\begin{math}c_0(u) = c_2(u)\end{math} for all~\begin{math}u \in  V(G)\end{math}. Suppose that~\begin{math}c_0\end{math} does not have property plus. There are two cases to consider:
		\begin{enumerate}
			\item There exists~\begin{math}uv \in E(G)\end{math} such that~\begin{math}c_0(u) < c_0(v)\end{math}, but~\begin{math}c_1(u) \leq c_1(v)\end{math}; or
			\item there exists~\begin{math}uv \in E(G)\end{math} such that~\begin{math}c_0(u) = c_0(v)\end{math}, but~\begin{math}c_1(u) < c_1(v)\end{math}.			
		\end{enumerate}
	
	Assume there exists an edge~\begin{math}u_0u_1\end{math} where~\begin{math}c_0(u_0)<c_0(u_1)\end{math} and~\begin{math}c_1(u_0)\leq c_1(u_1)\end{math}.
	Since~\begin{math}c_0(u_1) = c_2(u_1)\end{math}, there exists some edge~\begin{math}u_1u_2\end{math} so that~\begin{math}c_0(u_1) \leq c_0(u_2)\end{math} and~\begin{math}c_1(u_1) < c_1(u_2)\end{math}.
	Repeating this argument yields a sequence of vertices~\begin{math}u_0, u_1, u_2, u_3 \dots \end{math} so that \begin{equation}c_0(u_0) < c_0(u_1) \leq c_0(u_2) < c_0(u_3) \leq \dots.\end{equation}
	However, since~\begin{math}G\end{math} is finite, eventually there exists a pair of indices~\begin{math}i,j\end{math} so that~\begin{math}i < j+1\end{math} and~\begin{math}u_i = u_j\end{math}.  This is a contradiction as it implies~\begin{math}c_0(u_i) < c_0(u_j).\end{math}

 	Assume there exists an edge~\begin{math}u_0u_1\end{math} where~\begin{math}c_0(u_0)=c_0(u_1)\end{math} and~\begin{math}c_1(u_0) < c_1(u_1)\end{math}. Proceeding as in the previous case yields the same contradiction.
\end{proof}

Property plus describes the necessary and sufficient conditions for the process to be tight. And so any graph that has property plus with respect to~\begin{math}c_0\end{math} will either enter a cycle with period length~\begin{math}2\end{math} with no pre-period or will be fixed.

A \emph{ full-degree initial chip configuration} is an initial configuration of~\begin{math}G\end{math} where each vertex in a graph has degree-many chips.

We next prove that given a full-degree initial configuration, chip configurations on paths are eventually periodic with period length~\begin{math}2\end{math}.  We illustrate the first five chip configurations in Figure~\ref{Fig:PathPeriod2}.  One can easily observe the movement of chips propagates toward centre of the path.  

\begin{figure}[htbp]
\[ \includegraphics[width=0.9\textwidth]{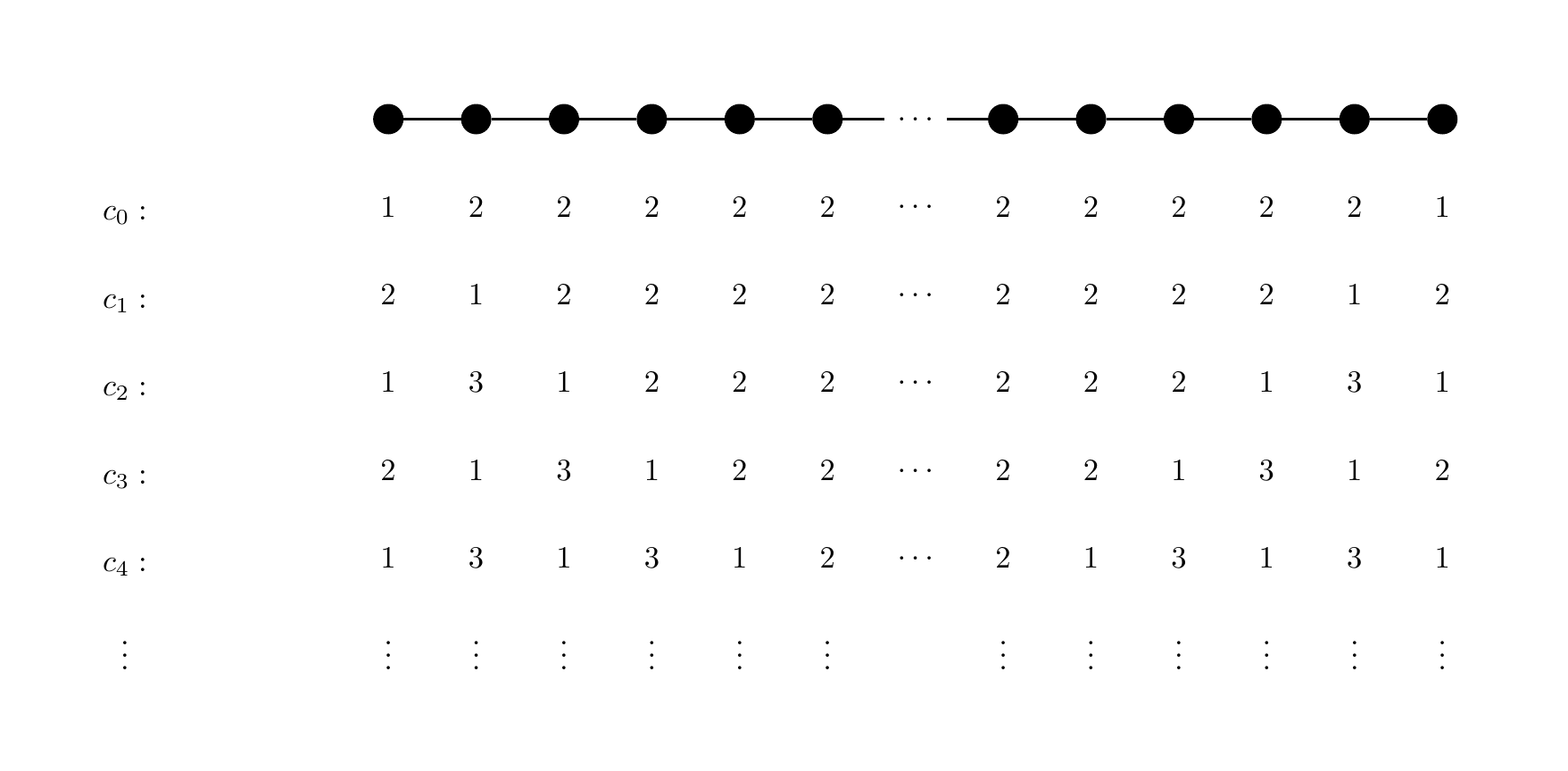} \]

\vspace{-0.25in}\caption{Configurations on a path $P_n$.}

\label{Fig:PathPeriod2}
\end{figure}

We begin by proving a lemma regarding the process on the one-way infinite path with a full-degree chip configuration. To do this we consider the sequence~\begin{math}\left\lbrace c_t(v_i)\right\rbrace_{i\geq1}\end{math} as an infinite word with entries drawn from the alphabet~\begin{math}\mathit{\{1,2,3\}}\end{math}. Recall that~\begin{math}\mathit{(ab)^k}\end{math} denotes the sequence of length~\begin{math}2k\end{math} consisting of alternating~\begin{math}\mathit{a}\end{math}'s and~\begin{math}\mathit{b}\end{math}'s.

\begin{lemma}\label{lem:infinitePath}
	Let~\begin{math}P_\infty = u_1,u_2, \dots\end{math} be the one way infinite path. For all~\begin{math}t = 2k\end{math}, we have~\begin{math}\left\lbrace c_{t}(u_i)\right\rbrace_{i\geq 1} = \mathit{(13)^k1222\dots}\end{math}. For all~\begin{math}t = 2k+1\end{math}, we have~\begin{math}\left\lbrace c_t(u_i)\right\rbrace_{i\geq 1} = \mathit{2(13)^k1222\dots}\end{math}.  
\end{lemma}

\begin{proof}
	We proceed by induction on~\begin{math}k\end{math}, noting that our claims are true for~\begin{math}k=0\end{math}. Assume~\begin{math}k=r>0\end{math} and consider~\begin{math}t = 2r\end{math}.
	By induction we have \begin{displaymath}\left\lbrace c_{2(r-1)+1}(u_i)\right\rbrace_{i\geq 1} = \mathit{2(13)^{r-1}1222\dots}.\end{displaymath} Observe	
	\begin{displaymath}\Delta_{2(r-1)+1}(u_i) = 
	\begin{cases}
    -2 & i \equiv 1 \bmod 2, \;  (3 \leq i \leq 2r-1)   \\
    -1 & i = 1, 2r+1 \\
    0 &  i \geq 2r+2\\ 
    2 & i \equiv 0 \bmod 2,  \; (2 \leq i \leq 2r). \\
	\end{cases}\end{displaymath}
	The result now follows by application of the process dynamics. A similar argument gives the case~\begin{math}t = 2r + 1\end{math}.
\end{proof}

\begin{theorem} The path~\begin{math}P_n\end{math} with full-degree initial configuration is periodic with period length~\begin{math}2\end{math} with a pre-period length of~\begin{math}rad(G)-1\end{math}.
\end{theorem}

\begin{proof}
	The result is trivially true for~\begin{math}n = 1,2\end{math}. Let~\begin{math}P_n = u_1,u_2, \dots u_n\end{math} be the path on~\begin{math}n\geq 3\end{math} vertices and let~\begin{math}c_0\end{math} be the full-degree configuration. Note that for all~\begin{math}n\geq 3\end{math}, we have~\begin{math}rad(G)-1 = \lceil \frac{n-3}{2}\rceil\end{math}. 
	Since the dynamics are entirely local, for fixed~\begin{math}n\geq 3\end{math} and~\begin{math}t < \lceil \frac{n-3}{2}  \rceil\end{math} the behaviour on each side of the path follows the behaviour of the process on the infinite one-way path with full-degree configuration. 
	For fixed~\begin{math}n\geq 3\end{math} and~\begin{math}t = \lceil\frac{n-3}{2}\rceil\end{math} we have the following four cases ($k \geq 0$) following from Lemma \ref{lem:infinitePath}.
	
\begin{center}
	\begin{tabular}{|c|c|c|}
		\hline 
		~\begin{math}\mathbf{n}\end{math}&~\begin{math}\mathbf{t}\end{math} &	~\begin{math}\mathbf{\left\lbrace c_{t}(u_i)\right\rbrace_{i=1}^{i=n}}\end{math}\\
		\hline 
		~\begin{math}4k+3\end{math} &~\begin{math}2k\end{math} &~\begin{math}\mathit{   (13)^k121(31)^k}\end{math}  \\ 
		\hline
		~\begin{math}4k+4\end{math} &~\begin{math}2k\end{math} &~\begin{math}\mathit{   (13)^k1221(31)^k}\end{math}  \\ 
		\hline
		~\begin{math}4k+5\end{math} &~\begin{math}2k+1\end{math} & ~\begin{math}\mathit{2(13)^k121(31)^k2}\end{math}\\
		\hline
		~\begin{math}4k+6\end{math} &~\begin{math}2k+1\end{math} & ~\begin{math}\mathit{2(13)^k1221(31)^k2}\end{math}\\
		\hline 
	\end{tabular} 
	\end{center}

In each of these cases it is easily checked that~\begin{math}c_{\lceil\frac{n-3}{2}\rceil} = c_{rad(G) -1}\end{math} has property plus. The result now follows from Theorem \ref{thm:propertyplus}.\end{proof}

Using property plus, we may also examine the behaviour of the process on the complete graph where each of the vertices has one of two initial values.

\begin{theorem}  On a complete graph~\begin{math}K_n\end{math} for~\begin{math}\alpha,\beta,d\in\mathbb{Z}\end{math} with~\begin{math}1\leq d\leq n\end{math}, if~\begin{math}d\end{math} of the vertices are initially assigned~\begin{math}\alpha\end{math}-many chips and the remaining~\begin{math}n-d\end{math} vertices are initially assigned~\begin{math}\beta\end{math}-many chips, then the configuration is tight.
\end{theorem}

\begin{proof}  
	This claim is trivially true for~\begin{math}n = 1,2\end{math}. Assume~\begin{math}n \geq 3\end{math}. Let~\begin{math}A \subseteq V(K_n)\end{math} be the set of vertices initially with~\begin{math}\alpha\end{math} chips.  Let~\begin{math}B = V(K_n) \backslash A\end{math} be the set of vertices initially with~\begin{math}\beta\end{math} chips. 
	
As~\begin{math}K_n\end{math} is a complete graph, every vertex has the same closed neighbourhood. By Corollary~\ref{cor:equal}, at each step, all vertices of~\begin{math}A\end{math} have the same number of chips and at each step, all vertices of~\begin{math}B\end{math}  have the same number of chips.

Let~\begin{math}a \in A\end{math} and~\begin{math}b \in B\end{math}. Observe that if~\begin{math}c_t(a) < c_t(b)\end{math}, then~\begin{math}\Delta_t(a) = |B|\end{math} and~\begin{math}\Delta_t(b) = -|A|\end{math} . From this we see that if~\begin{math}c_t(a) < c_t(b)\end{math} and~\begin{math}c_{t+1}(a) < c_{t+1}(b)\end{math}, then \begin{equation}0 < c_{t+1}(b) - c_{t+1}(a)  =  c_t(b) -|A | - c_t(a) - |B| <c_t(b) - c_t(a).\end{equation} 

From this it follows that there exists~\begin{math}t\end{math} such that~\begin{math}c_t(a) < c_t(b)\end{math} and~\begin{math}c_{t+1}(a) \geq c_{t+1}(b)\end{math}. If this relationship holds with equality, then~\begin{math}c_{t+1}\end{math} is fixed. Otherwise, observe that for all~\begin{math}a \in A\end{math} and~\begin{math}b \in B\end{math} we have~\begin{math}c_t(b) < c_t(a)\end{math} and~\begin{math}c_{t+1}(b) > c_{t+1}(a)\end{math}. The result holds by Theorem \ref{thm:propertyplus}.\end{proof}

Note that the above result will also hold for complete bipartite graphs where all vertices in one partite set initially receive~\begin{math}\beta\end{math} chips and the vertices in the other partite set receive~\begin{math}\alpha\end{math} chips.


\subsection{Periodicity of Stars}


We refer to the complete bipartite graph~\begin{math}K_{1,n-1}\end{math} for~\begin{math}n \geq 1\end{math} as a \emph{ star on $n $ vertices} (denoted~\begin{math}S_n\end{math}).  In this subsection we consider~\begin{math}S_n\end{math} on~\begin{math}n \geq 1\end{math} vertices.
We  show that given any finite initial chip configuration, the chip configuration is eventually tight. 

Denote the centre vertex of~\begin{math}S_n\end{math} by~\begin{math}v\end{math}.  For any finite initial configuration of chips, let~\begin{math}L_M\end{math} and~\begin{math}L_m\end{math} denote the set of leaves with the maximum and minimum number of chips in the initial configuration, respectively. 

By Corollary \ref{cor:equal} we have~\begin{math}c_t(\ell_1) = c_t(\ell_2)\end{math} for all~\begin{math}\ell_1, \ell_2\in L_M\end{math} (or~\begin{math}\ell_1, \ell_2\in L_m\end{math}) and all~\begin{math}t \geq 0\end{math}. As such, let~\begin{math}\ell_M\end{math}  denote any vertex in~\begin{math}L_M\end{math} and,~\begin{math}\ell_m\end{math}  denote any vertex in~\begin{math}L_m\end{math}.  Let \begin{displaymath}d_t = \max_{x,y \in V(S_n)\backslash\{v\}} |c_t(x) - c_t(y)|,\end{displaymath} the maximum difference between the number of chips at any two leaves at step~\begin{math}t\end{math}.

\begin{lemma}\label{lem:starLem}
	Let~\begin{math}c_0\end{math} be a finite chip configuration on~\begin{math}S_n\end{math}, with~\begin{math}n \geq 3\end{math}.
	
	\begin{enumerate}
		\item For all~\begin{math}t \geq 0\end{math}, we have~\begin{math}d_t = c_t(\ell_M) - c_t(\ell_m)\end{math}.
		\item If~\begin{math}c_t(\ell_m) \leq  c_t(v) \leq c_t(\ell_M)\end{math}  and one of these inequalities is strict, then~\begin{math}d_{t+1} < d_t\end{math}.
		\item If~\begin{math}c_t(v) < c_t(\ell_m)\end{math}, then there exists~\begin{math}r\geq t\end{math} such that~\begin{math}c_r(v) \geq c_r(\ell_m)\end{math}.
		\item If~\begin{math}c_t(v) <  c_t(\ell_m)\end{math} and~\begin{math}c_{t+1}(v) > c_{t+1}(\ell_M)\end{math}, then~\begin{math}c_t\end{math} has property plus.
		\item If~\begin{math}d_t=0\end{math}, then there exists~\begin{math}r \geq t\end{math} such that~\begin{math}c_r\end{math} has property plus. 
	\end{enumerate}
	
\end{lemma}

\begin{proof}
	Let~\begin{math}c_0\end{math} be a finite chip configuration on~\begin{math}S_n\end{math}.
	
	\begin{enumerate}
		\item We show for all~\begin{math}\ell_1,\ell_2 \neq v\end{math} that if~\begin{math}c_t(\ell_1) \leq c_t(\ell_2)\end{math}, then~\begin{math}c_{t+1}(\ell_1) \leq c_{t+1}(\ell_2)\end{math}.  
		If~\begin{math}c_t(\ell_2)  < c_t(v)\end{math}, then~\begin{math}c_{t+1}(\ell_1) = c_t(\ell_1) + 1\end{math} and~\begin{math}c_{t+1}(\ell_2) = c_t(\ell_2) + 1\end{math} and the result holds. Similarly, the result holds if~\begin{math}c_t(\ell_1)  > c_t(v)\end{math}.
		Consider~\begin{math}c_t(\ell_1) \leq c_t(v) \leq c_t(\ell_2)\end{math}. 
		If neither of these equalities is strict, then the result holds as~\begin{math}c_t(\ell_1) = c_{t+1}(\ell_1)\end{math} and~\begin{math}c_t(\ell_2) = c_{t+1}(\ell_2)\end{math}.
		If~\begin{math}c_t(\ell_1) < c_t(v) \leq c_t(\ell_2)\end{math}, then~\begin{math}c_{t+1}(\ell_1) = c_{t}(\ell_1)+1\end{math} and~\begin{math}c_{t+1}(\ell_2) \geq c_{t+1}(\ell_2) -1\end{math}, and the result holds. 
		Similarly if~\begin{math}c_t(\ell_1) \leq c_t(v) < c_t(\ell_2)\end{math}, the result holds.
		
		\item Without loss of generality, assume~\begin{math}c_t(\ell_m) <  c_t(v) \leq c_t(\ell_M)\end{math}. In this case we have~\begin{math}\Delta_t(\ell_m) = 1\end{math} and~\begin{math}\Delta_t(\ell_M) \leq 0\end{math}. The result now follows from part 1.
		
		\item  If~\begin{math}c_t(v) < c_t(\ell_m)\end{math}, then~\begin{math}c_t(u) > c_t(v)\end{math} for all~\begin{math}u \neq v\end{math},~\begin{math}\Delta_t(u)=-1\end{math} for all~\begin{math}u \neq v\end{math} and~\begin{math}\Delta_{t}(v) = deg(v) > 1\end{math}.
		Therefore~\begin{math}c_{t+1}(\ell_m) - c_{t+1}(v) < c_t(\ell_m) - c_t(v)\end{math}. 
		Since~\begin{math}c_t(\ell_m) - c_t(v) > 0\end{math} and all values are integral, repeating this argument eventually yields some~\begin{math}i \geq 1\end{math} so that~\begin{math}c_{t+i}(\ell_m) - c_{t+i}(v) \leq 0\end{math}. Setting~\begin{math}r = t+i\end{math} gives the desired result.
		
		\item If~\begin{math}c_t(v) < c_t(\ell_m)\end{math}, then~\begin{math}c_t(u) > c_t(v)\end{math} for all~\begin{math}u \neq v\end{math}, and~\begin{math}\Delta_t(u)=-1\end{math} for all~\begin{math}u \neq v\end{math}. 
		Therefore~\begin{math}c_t(x)+\Delta_t(x) > c_t(y)+\Delta_t(y)\end{math} holds for all~\begin{math}xy \in E(G)\end{math} such that~\begin{math}c_t(x)~<~c_t(y)\end{math}. 
		If~\begin{math}c_{t+1}(v) > c_{t+1}(\ell_M)\end{math}, then~\begin{math}c_{t+1}(u) > c_{t+1}(v)\end{math} for all~\begin{math}u \neq v\end{math} and a similar reasoning applies, therefore~\begin{math}c_t\end{math} has property plus
	
		\item Assume~\begin{math}d_t=0\end{math}. If~\begin{math}c_t(v) = c_t(\ell_m)\end{math}, then~\begin{math}c_t(u) = c_t(v)\end{math} for all~\begin{math}v \in V(G)\end{math}.  In this case we note that~\begin{math}\Delta_t(x) = 0\end{math} for all~\begin{math}x \in V(G)\end{math} and so~\begin{math}c_t\end{math} has property plus. Otherwise, we may assume without loss of generality that~\begin{math}c_t(v) < c_t(\ell_m)= c_t(\ell_M)\end{math}. By part 3 there exists~\begin{math}r\end{math} such that~\begin{math}c_r(v) \geq c_r(\ell_m) = c_t(\ell_M)\end{math}. If this inequality is strict then our claim holds by part 4. Otherwise our claim holds by the previous case in part 5.\end{enumerate}\end{proof}

\begin{theorem}
	For any finite initial chip configuration on~\begin{math}S_n\end{math}, the chip configuration is tight with pre-period length at most   \begin{displaymath}\left\lceil\frac{\max \left\lbrace 0, c_0(v) - c_0(\ell_M), c_0(\ell_m) - c_0(v) \right \rbrace}{n}\right\rceil + 2d_0.\end{displaymath}
\end{theorem}

\begin{proof}
	Let~\begin{math}c_0\end{math} be an initial chip configuration on~\begin{math}S_n\end{math} for~\begin{math}n \geq 1\end{math}. We note that the claim is trivially true for~\begin{math}n \leq 2\end{math}. 
	Thus we proceed and assume~\begin{math}n \geq 3\end{math}. 
	
	If~\begin{math}c_0(v) <  c_0(\ell_m)\end{math}, then by part 3 of Lemma \ref{lem:starLem}, there exists a least~\begin{math}t>0\end{math} such that~\begin{math}c_{t-1}(v)~<~c_{t-1}(\ell_m)\end{math} and~\begin{math}c_{t}(v) \geq c_{t}(\ell_m)\end{math}. 
	Observe that if~\begin{math}c_i(v) < c_i(\ell_m)\end{math}, then~\begin{math}c_{i+1}(v) = c_i(v) + deg(v)\end{math} and~\begin{math}c_{i+1}(u) = c_i(u) -1\end{math} for all~\begin{math}u \neq v\end{math}. 
	Therefore~\begin{math}t \leq \lceil\frac{c_0(\ell_m) - c_0(v)}{deg(v)+1}\rceil\end{math}.
	Similarly, if~\begin{math}c_0(v) >  c_0(\ell_M)\end{math}, then there exists~\begin{math}t\end{math} such that~\begin{math}c_{t-1}(v) > c_{t-1}(\ell_M)\end{math} and~\begin{math}c_{t}(v) \leq c_{t}(\ell_M)\end{math}, where~\begin{math}t \leq \lceil\frac{c_0(v) -c_0(\ell_M) }{deg(v)+1}\rceil\end{math}.
	
	Since~\begin{math}c_i(\ell_m) \leq c_i(\ell_M)\end{math} for all~\begin{math}i \geq 0\end{math}, there exists~\begin{math}t \leq \left\lceil\frac{\max \left\lbrace 0, c_0(v) - c_0(\ell_M), c_0(\ell_m) - c_0(v) \right \rbrace}{deg(v) +1}\right\rceil\end{math} such that~\begin{math}c_t(\ell_m) \leq c_t(v) \leq c_t(\ell_M)\end{math}.
	If neither of these equalities is strict, then~\begin{math}d_t=0\end{math} and by part 5 of Lemma \ref{lem:starLem},~\begin{math}c_t\end{math} has property plus.
	
	Otherwise, assume without loss of generality that~\begin{math}c_t(v) < c_t(\ell_M)\end{math}. By part 2 of Lemma \ref{lem:starLem} we have~\begin{math}d_{t+1} < d_t\end{math}. 
	If~\begin{math}c_{t+1}(v) > c_{t+1}(\ell_M)\end{math}, then~\begin{math}c_{t+1}(v) - c_{t+1}(\ell_M) < deg(v)\end{math}. 
	Therefore~\begin{math}c_{t+2}(v) \leq c_{t+2}(\ell_M)\end{math}. 
	If~\begin{math}c_{t+2}(v) < c_{t+2}(\ell_m)\end{math}, then by part 4 of Lemma \ref{lem:starLem}~\begin{math}c_{t+2}\end{math} has property plus. 
	Otherwise~\begin{math}c_{t+2}(\ell_m) \leq c_{t+2}(v) \leq c_{t+2}(\ell_M)\end{math} and~\begin{math}d_{t+2} = d_{t+1}\end{math}.
	Therefore if~\begin{math}c_t(\ell_m) \leq c_t(v) \leq c_t(\ell_M)\end{math}, then either~\begin{math}c_t\end{math} has property plus,~\begin{math}c_{t+2}\end{math} has property plus, or~\begin{math}d_{t+2} \leq d_t-1\end{math}. Recall from part 5 of Lemma \ref{lem:starLem} that if~\begin{math}d_i=0\end{math}, then~\begin{math}c_i\end{math} has property plus.
	Since~\begin{math}d_t \leq d_0\end{math} for all~\begin{math}t > 0\end{math}, we have that if~\begin{math}c_t(\ell_m) \leq c_t(v) \leq c_t(\ell_M)\end{math}, then the process is periodic after at most an additional~\begin{math}2d_0\end{math} steps.

	Therefore 	for any finite initial chip configuration on~\begin{math}S_n\end{math}, the chip configuration is tight with pre-period length at most  
	\begin{equation}\left\lceil\frac{\max \left\lbrace 0, c_0(v) - c_0(\ell_M), c_0(\ell_m) - c_0(v) \right \rbrace}{deg(v) +1}\right\rceil + 2d_0.\end{equation}
	\end{proof}

\subsection{Mill-pond configurations}

A  configuration in which there exists~\begin{math}v \in V(G)\end{math} such that~\begin{math}c_0(v)=1\end{math} and for all vertices~\begin{math}x\neq v\end{math},~\begin{math}c_0(x)=0\end{math} is called a \emph{mill-pond}\footnote{We are dropping a single chip into a perfectly flat pond and watching what happens to the waves.}, denote this by~\begin{math}
MP(v) \end{math}.
Recall that for a graph~\begin{math}G\end{math} and a vertex~\begin{math}u \in V(G)\end{math}, the \emph{eccentricity of~\begin{math}u\end{math}}, denoted~\begin{math}\epsilon(u)\end{math}, is given by \begin{displaymath}\epsilon(u) = \max\{d(u,v)| v \in V(G)\},\end{displaymath} where~\begin{math}d(u,v)\end{math} is the length of the shortest path from~\begin{math}u\end{math} to~\begin{math}v\end{math}.
 We show that if~\begin{math}G\end{math} is bipartite, starting 
 with~\begin{math}
 MP(v)
 \end{math} for any~\begin{math}v\end{math}, then the process is periodic with period length~\begin{math}2\end{math} after a pre-period of length
~\begin{math}\epsilon(v)\end{math}. In particular,
the pre-period is at most the diameter of~\begin{math}G\end{math}.

For~\begin{math}i\geq 0\end{math}, let~\begin{math}N_i(v)\end{math} be the set of vertices at distance~\begin{math}i\end{math} from~\begin{math}v\end{math}.   For~\begin{math}x\in N_i(v)\end{math}, where~\begin{math}i\geq 1\end{math}, set
\begin{math}deg^+(x)=|N(x)\cap N_{i-1}(v)|\end{math}. Similarly for~\begin{math}x\in N_i(v)\end{math}, where~\begin{math}i\geq 0\end{math}, set~\begin{math}deg^-(x)=|N(x)\cap N_{i+1}(v)|\end{math}. Note that we define~\begin{math}N_0(v) = \{v\}\end{math} and also note that~\begin{math}deg^+(v) = 0\end{math}.

\begin{theorem}\label{thm:millpond} Let~\begin{math}G\end{math} be a connected bipartite
 graph,~\begin{math}v\in V(G)\end{math} and~\begin{math}c_0 = MP(v)\end{math}.  
  If~\begin{math}t\end{math} is even, then~\begin{math}c_t(v) =1\end{math}, otherwise~\begin{math}c_t(v) = 1- deg(v)\end{math}. 
Further for~\begin{math}x\in N_i(v)\end{math} and~\begin{math}i=1,2,\ldots,\epsilon(v)\end{math}, we have

\begin{displaymath}
c_t(x)=
\begin{cases} 0& \mbox{\begin{math}t<i\end{math}}\\
deg^+(x)& \mbox{\begin{math}t\geq i\end{math} and~\begin{math}t-i\equiv 0 \mod 2\end{math}}     \\
-deg^-(x)&\mbox{\begin{math}t\geq i\end{math} and~\begin{math}t-i\equiv 1 \mod 2\end{math}} 
\end{cases}
\end{displaymath}

\end{theorem}

\begin{proof}  Let~\begin{math}G\end{math} be a connected bipartite graph,~\begin{math}v\in V(G)\end{math} 
and~\begin{math}c_0\end{math} be $MP(v)$.   Note that for a bipartite graph,~\begin{math}deg(x) = deg^+(x)+deg^-(x)\end{math}.
We proceed by induction on~\begin{math}t\end{math}, noting that our claim is true for~\begin{math}t=0\end{math} by the definition of a mill-pond configuration.

First consider~\begin{math}v\end{math}. If~\begin{math}t\end{math} is odd then, by induction,~\begin{math}c_{t-1}(v) = 1\end{math} and, for~\begin{math}x\in N(v)\end{math},~\begin{math}c_t(x)\leq 0\end{math} (note that ~\begin{math}c_t(x)=0\end{math} if~\begin{math}t=1\end{math} and~\begin{math}c_t(x) = -deg^-(x)\end{math} otherwise)
 and so~\begin{math}c_{t}(v) = 1-deg(v)\end{math}. If~\begin{math}t\end{math} is even
then~\begin{math}c_{t-1}(v) = 1-deg(v)\leq 0\end{math} and, for~\begin{math}x\in N(v)\end{math},~\begin{math}c_{t-1}(x) = deg^+(x)> 0\end{math}. Thus~\begin{math}c_{t}(v) = 1\end{math}.

Consider~\begin{math}x_i\in N_i(v)\end{math},~\begin{math}i>0\end{math}. If~\begin{math}t<i\end{math}, then clearly~\begin{math}c_t(x_i) = 0\end{math} . Assume~\begin{math}t\geq i\end{math}.
In the case~\begin{math}t=i\end{math}, all vertices in~\begin{math}N_{i-1}(v)\end{math} have at least one chip at time~\begin{math}t-1\end{math}
and all vertices in~\begin{math}N_{j}(v)\end{math},~\begin{math}j\geq i\end{math}, have none. Therefore,~\begin{math}c_t(x_i) = deg^+(x_i)\end{math}.

We may now assume~\begin{math}t-i>0\end{math}. 

First suppose~\begin{math}t-i\end{math} is even.
By induction,~\begin{math}c_{t-1}(x_i) = -deg^-(x_i)<0\end{math}.
Consider~\begin{math}y\sim x_i\end{math}. If~\begin{math}y\in N_{i-1}\end{math}, then~\begin{math}c_{t-1}(y) = deg^+(y)>0\end{math}.   If~\begin{math}y\in N_{i+1}\end{math} 
then~\begin{math}c_{t-1}(y)= deg^+(y)\end{math} or~\begin{math}c_{t-1}(y)=0\end{math} if~\begin{math}t-i=1\end{math}.  In either case,~\begin{math}c_{t-1}(y)\geq 0\end{math}. Thus~\begin{math}c_{t}(x_i) = -deg^-(x_i) +deg(x_i) = deg^+(x_i)\end{math}.

The argument  for~\begin{math}t-i\end{math} odd follows similarly.

\end{proof}

Table \ref{table:seq} gives the first few configurations for a mill-pond configuration,  
on a bipartite graph where~\begin{math}v\end{math} has eccentricity,~\begin{math}\epsilon(v)\end{math}, at least~\begin{math}5\end{math}.

\begin{table}
\begin{tabular}{c|ccccccc}
~\begin{math}t\end{math}&~\begin{math}v\end{math}&~\begin{math}x_1\in N_1\end{math}&~\begin{math}x_2\in N_2\end{math}&~\begin{math}x_3\in N_3\end{math}&~\begin{math}x_4\in N_4\end{math}&~\begin{math}x_5\in N_5\end{math}\\
\hline
0&1&0&0&0&0&0\\
1&~\begin{math}1-deg(v)\end{math}&1&0&0&0&0\\
2&1&~\begin{math}1-deg(x_1)\end{math}&~\begin{math}\deg^+(x_2)\end{math}&0&0&0\\
3&~\begin{math}1-deg(v)\end{math}&1&~\begin{math}-deg^-(x_2)\end{math}&~\begin{math}deg^+(x_3)\end{math}&0&0\\
4&1&~\begin{math}1-deg(x_1)\end{math}&~\begin{math}deg^+(x_2)\end{math}&~\begin{math}deg^-(x_3)\end{math}&~\begin{math}deg^+(x_4)\end{math}&0\\
\end{tabular}
\caption{Configurations for the first $4$time steps.}
\label{table:seq}
\end{table}

After~\begin{math}t=\epsilon(v)\end{math} steps, the arguments for all the cases still hold, except no vertices, other than possibly~\begin{math}v\end{math}, have~\begin{math}0\end{math} chips. This proves the following. 

\begin{corollary}\label{cor:unichip} Let~\begin{math}G\end{math} be a  
bipartite
graph,~\begin{math}v\in V(G)\end{math} and~\begin{math}c_0\end{math} is $MP(v)$.  The configurations become periodic, with period length~\begin{math}2\end{math}, after a preperiod of length~\begin{math}\epsilon(v)-1\end{math}.
\end{corollary}

We note that Corollary \ref{cor:unichip} does not hold in general. There exist examples on as few as~\begin{math}6\end{math} vertices of non-bipartite graphs with a mill-pond configuration in which the pre-period length exceeds~\begin{math}\epsilon(v)-1\end{math}. For example, the mill-pond configuration for the graph given in Figure \ref{Fig:UniChip} with~\begin{math}c_0(v) = 1\end{math} has pre-period length~\begin{math}6\end{math}; however,~\begin{math}\epsilon(v)-1 = 1\end{math}.

\begin{figure}[htbp]
	\[ \includegraphics[width=0.25\textwidth]{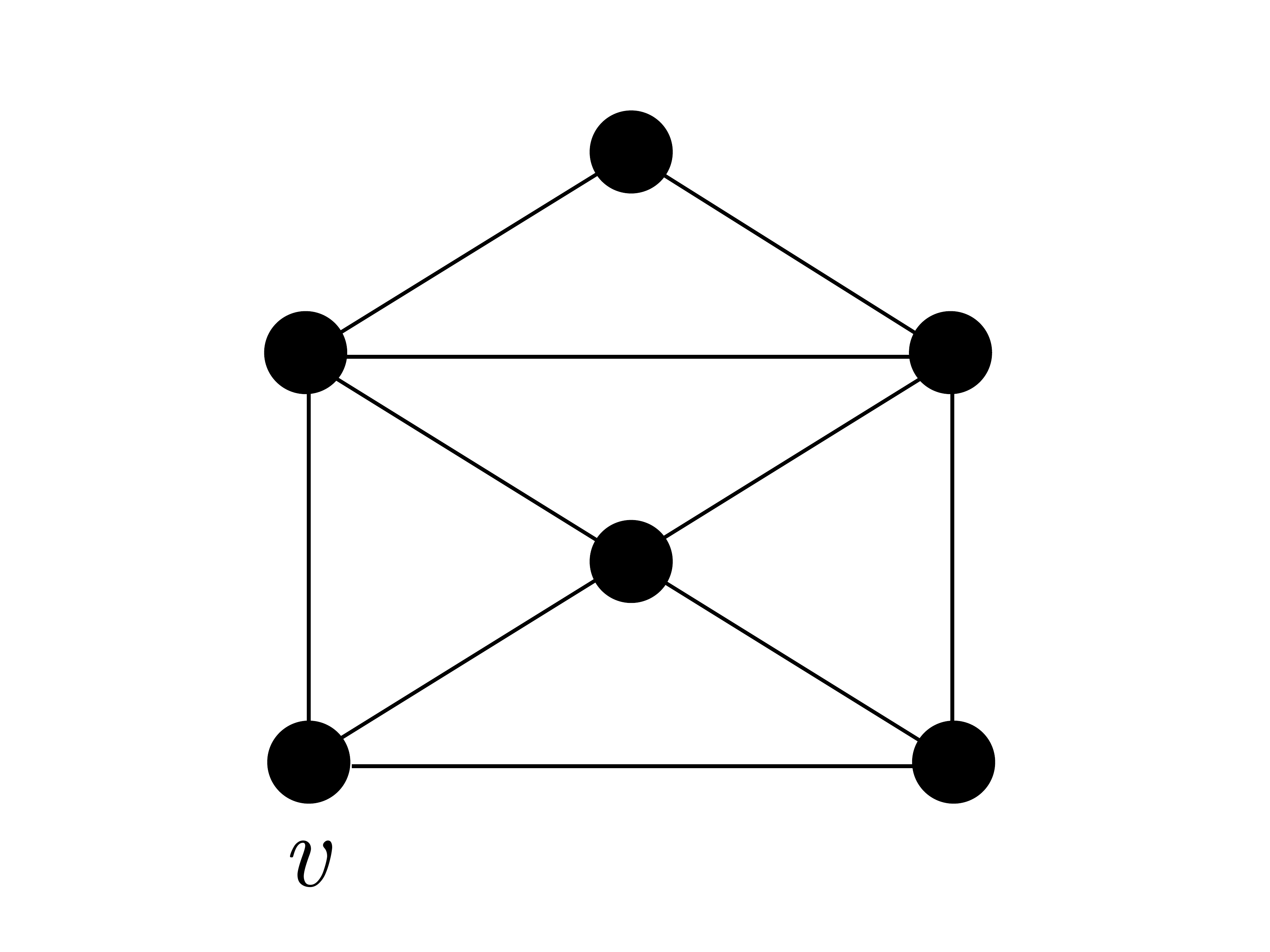} \]
	
	\vspace{-0.25in}\caption{A mill-pond configuration with pre-period length not equal to $\epsilon(v) -1$.}
	
	\label{Fig:UniChip}
\end{figure}


\section{Questions}\label{sec:questions}


Although we have determined that configurations are eventually periodic for some classes of graphs, we restate and comment on Questions~\ref{q1} and Conjecture~\ref{conj}.\\

\noindent {\bf Question~\ref{q1}}~\emph{ Let~\begin{math}G\end{math} be a finite  graph with an initial configuration.  Do the chip configurations on~\begin{math}G\end{math} eventually exhibit periodic behaviour? }\\

\noindent {\bf Conjecture~\ref{conj}}~\emph{ Every finite graph and initial configuration is tight.}\\

Although we have answered Question~\ref{q1} affirmatively for paths, cycles, wheels, stars, complete graphs, and complete bipartite graphs, the question remains open for general graphs. Conjecture~\ref{conj} was fully resolved for stars and some partial results were provided in Section~\ref{sec:period2}, but the conjecture remains open. By Theorem \ref{thm:propertyplus}, it is enough to show that for every graph and every initial chip configuration, the configuration eventually attains property plus.\\ 

If an initial configuration is known to be eventually periodic, then the initial configuration can be offset by a constant factor so that no vertex ever has a negative number of chips. This observation yields the following question.
  
\begin{question} Given~\begin{math}k\end{math} chips to be distributed on vertices of a graph~\begin{math}G\end{math} (and assuming eventual periodicity), how many different configurations will result in no vertex ever having a negative number of chips? 
\end{question} 

In the case that~\begin{math}G\end{math} is a finite bipartite graph in which each vertex begins with~\begin{math}0\end{math} chips, except for a single vertex,~\begin{math}v\end{math}, that begins with~\begin{math}1\end{math} chip, we have shown that the process is periodic with period length~\begin{math}2\end{math} after a pre-period of length at most~\begin{math}\epsilon(v)-1\end{math}. In periodic configurations, each vertex, other than~\begin{math}v\end{math},  either always has value~\begin{math}0\end{math} or oscillates between having positive value and negative value.  This leads to the following  questions.

\begin{question} Let~\begin{math}G\end{math} be a graph with a finite initial chip configuration.  Should the configurations eventually exhibit periodic behaviour, how many steps will it take for the process to become periodic? \end{question}

For a graph~\begin{math}G\end{math} and a fixed vertex~\begin{math}v\end{math}, consider the configuration~\begin{math}c_0(v) = -deg(v)\end{math}, and for~\begin{math}x\ne v\end{math},
~\begin{math}c_0(x) = 1\end{math} if~\begin{math}x\sim y\end{math}, and~\begin{math}c_0(x) = 0\end{math} otherwise. Such a configuration arises from the zero-configuration (\textit{i.e.,} the configuration where the value at each vertex is~\begin{math}0\end{math}) by initially allowing~\begin{math}v\end{math} to send one chip to each of its neighbours.
 Denote this by~\begin{math}QF(v)\end{math}.  
 If~\begin{math}N[v]=V(G)\end{math} and~\begin{math}|V(G)|=n\end{math} then at~\begin{math}t=1\end{math},~\begin{math}c_0(x) =0\end{math} for all~\begin{math}x\end{math} and~\begin{math}QF(v)\end{math} has turned into a zero-configuration. We ask when else does this occur.
 
The~\begin{math}QF(v)\end{math} configuration is just one chip 
different from the mill-pond configuration. For a bipartite graph,
 if~\begin{math}\epsilon(v)>1\end{math}, then it is easy to see that the statement of Theorem \ref{thm:millpond} also holds for~\begin{math}QF(v)\end{math}  except for the number of chips at~\begin{math}v\end{math} is always~\begin{math}1\end{math} fewer.
 
\begin{question}
Characterize those graphs~\begin{math}G\end{math} which have a vertex~\begin{math}v \end{math} such that~\begin{math} QF(v)\end{math} eventually becomes a zero-configuration.
\end{question}

\bibliographystyle{abbrvnat}
\bibliography{chipFireBib}

\end{document}